\documentclass[letterpaper, 10 pt, conference]{ieeeconf}  % Comment this line out if you need a4paper

\IEEEoverridecommandlockouts                              % This command is only needed if 
\usepackage{hyperref}
\usepackage{cite}
\usepackage{amsmath,amssymb,amsfonts}
\usepackage{algorithmic}
\usepackage{graphicx}
\usepackage{textcomp}

\usepackage{subcaption}
\usepackage{multicol}

\newtheorem{remark}{Remark}
\newtheorem{theorem}{Theorem}
\newtheorem{definition}{Definition}
\DeclareMathOperator*{\argmin}{arg\,min}

\usepackage{tikz}
\newcommand\copyrighttext{%
  \footnotesize \textcopyright 2023 IEEE. Personal use of this material is permitted.
  Permission from IEEE must be obtained for all other uses, in any current or future
  media, including reprinting/republishing this material for advertising or promotional
  purposes, creating new collective works, for resale or redistribution to servers or
  lists, or reuse of any copyrighted component of this work in other works.
}
\newcommand\copyrightnotice{%
\begin{tikzpicture}[remember picture,overlay]
\node[anchor=south,yshift=10pt] at (current page.south) {\fbox{\parbox{\dimexpr\textwidth-\fboxsep-\fboxrule\relax}{\copyrighttext}}};
\end{tikzpicture}%
}

\title{\LARGE \bf
Safe Control Synthesis for Multicopter via Control Barrier Function Backstepping
}

\author{Jinrae Kim, and Youdan Kim, \IEEEmembership{Senior Member, IEEE}
\thanks{
  This research was supported by Unmanned Vehicles Core Technology Research and Development Program through the National Research Foundation of Korea (NRF) and Unmanned Vehicle Advanced Research Center (UVARC) funded by the Ministry of Science and ICT, the Republic of Korea (2020M3C1C1A01083162).
}
\thanks{
Jinrae Kim and Youdan Kim are with the Department of Aerospace Engineering, Seoul National University, Institute of Advanced Aerospace Technology, Seoul 08826, Republic of Korea
{\tt\small \{kjl950403,ydkim\}@snu.ac.kr}%
}
}

\begin{document}

\maketitle
\thispagestyle{empty}
\pagestyle{empty}
\copyrightnotice

%%%%%%%%%%%%%%%%%%%%%%%%%%%%%%%%%%%%%%%%%%%%%%%%%%%%%%%%%%%%%%%%%%%%%%%%%%%%%%%%
\begin{abstract}
  A safe controller for multicopter is proposed using control barrier function.
  Multicopter dynamics are reformulated to deal with mixed-relative-degree and non-strict-feedback-form dynamics,
  and a time-varying safe backstepping controller is designed.
  Despite the time-varying variation,
  it is proven that the control input can be obtained by solving quadratic programming with affine inequality constraints.
  The proposed controller does not utilize a cascade control system design,
  unlike existing studies on the safe control of multicopter.
  Various safety constraints on angular velocity, total thrust direction, velocity, and position
  can be considered.
  Numerical simulation results support that the proposed safe controller does not violate
  all safety constraints including low- and high-level dynamics.
\end{abstract}

%%%%%%%%%%%%%%%%%%%%%%%%%%%%%%%%%%%%%%%%%%%%%%%%%%%%%%%%%%%%%%%%%%%%%%%%%%%%%%%%
\section{Introduction}
\label{sec:introduction}
Safety is the most important requirement in many physical systems.
Safety requirements are typically captured as state constraints
such that the state of interest does not escape prescribed safe sets~\cite{amesControlBarrierFunctions2019}.
For example,
a geofence can be considered as a designated safe area in which multicopters fly.
In this regard,
control barrier function (CBF) has attracted a lot of attention
for safe control in various fields including robotics and aerospace engineering.
The CBF method is a modern interpretation of Nagumo's theorem for control system,
a principle result in viability theory~\cite{aubinViabilityTheory2009}.
Ames et al. drove sufficient and necessary conditions \textit{not just at the boundary of but also inside a safe set}
by incorporating class $\mathcal{K}$-like functions for the forward invariance of the safe set~\cite{amesControlBarrierFunctions2019}.
This observation triggered a new paradigm of optimization-based safety-critical controller design.

% TODO: multicopter and safety constraint, existing studies and challenges
Studies on safe control of multicopter include
i) model predictive control (MPC)~\cite{bemporadHierarchicalHybridModel2009},
% TODO
ii) barrier Lyapunov function (BLF)~\cite{liFinitetimeControlQuadrotor2021},
and iii) CBF with cascade controller design~\cite{khanBarrierFunctionsCascaded2020}.
MPC deals with state constraints during multi-step time window by solving an optimization problem.
However, MPC needs to perform linearization to make the optimization problem quadratic programming (QP)
or solve computationally burdening nonlinear optimization,
which may be inappropriate for multicopters due to high nonlinearity in dynamics.
BLF is a special form of Lyapunov function for stability and state constraints.
BLF is more restrictive than CBF in terms of
i) positive definiteness~\cite{amesControlBarrierFunctions2019}
and ii) types of safety (for example, only error bounds can be considered for tracking control).
On the other hand,
CBF is much less restrictive compared to BLF
and only requires QP for nonlinear input-affine systems~\cite{amesControlBarrierFunctions2019}.
In~\cite{khanBarrierFunctionsCascaded2020},
the multicopter dynamics are separated into low- and high-level dynamics,
and (high-order) CBFs are utilized in each hierarchy for safety with respect to velocity and position.
This approach, however, may violate the high-level safety due to the cascade design.
The challenges of CBF methods for multicopter come from
\textit{mixed-relative-degree} and \textit{non-strict-feedback-form dynamics}.

% TODO: proposed method and its characteristics
To overcome the challenges, in this study,
the multicopter dynamics are reformulated as a strict feedback form,
inspired by Lyapunov backstepping position control for multicopter~\cite{falconiAdaptiveFaultTolerant2016}.
Then, a CBF-based safe controller is proposed for multicopter based on safe backstepping and CBF methods.
Various safety constraints are considered in this study,
with respect to
i) angular velocity (rotate not too fast),
ii) total thrust direction (prevent inverted flight),
iii) velocity (move not too fast),
and iv) position (geofence or obstacle avoidance).
Based on a previous study~\cite{kimSafeAttitudeController2023},
safety with respect to angular velocity and total thrust direction is guaranteed by using conventional and high-order CBFs~\cite{amesControlBarrierFunctions2019,xiaoHighOrderControlBarrier2022}.
Safety with respect to velocity and position is satisfied by safe backstepping with the reformulation~\cite{taylorSafeBacksteppingControl2022}.
Numerical simulation results support that the proposed safe controller does not violate any safety constraints considered in this study,
while the existing cascade safe controller for multicopter violates high-level safety as reported in~\cite{khanBarrierFunctionsCascaded2020}.

% TODO: organization
The rest of this paper is organized as follows.
\autoref{sec:preliminaries},
introduces the preliminaries of this study including
control barrier function method, its variants,
and the dynamics of multicopter.
In \autoref{sec:main_results},
the multicopter dynamics are reformulated to incorporate the safe backstepping method,
and CBFs are driven for multiple safety constraints.
It is shown that the optimization problem of the proposed safe controller is a QP with
affine inequality constraints.
In \autoref{sec:numerical_simulation},
numerical simulation is performed to demonstrate the performance of the proposed controller.
\autoref{sec:conclusion} concludes this study with future works.

\section{Preliminaries}
\label{sec:preliminaries}
% The sets of all, nonnegative, and positive real numbers are denoted by
% $\mathbb{R}$,
% $\mathbb{R}_{\geq 0}$,
% and $\mathbb{R}_{> 0}$,
% respectively.
% The standard Eucleadean norm is used in this study, i.e., $\lVert x \rVert := \sqrt{x^{\intercal} x}$ for any $x \in \mathbb{R}^{n}$.
$\mathcal{K}_{\infty}$ denotes the set of all strictly increasing continuous functions $\alpha: [0, a) \to [0, \infty)$
such that $\alpha(0) = 0$,
and $\mathcal{K}_{\infty}^{e}$ denotes the extended class $\mathcal{K}_{\infty}$,
that is,
for any $\alpha \in \mathcal{K}_{\infty}^{e}$,
$\alpha $ is in $\mathcal{K}_{\infty}$ and $\lim_{r \to -\infty} \alpha(r) = -\infty$~\cite{taylorSafeBacksteppingControl2022}.

Consider a nonlinear dynamical system,
\begin{equation}
  \label{eq:autonomous_system}
  \dot{x} = f(x),
\end{equation}
with state $x \in \mathbb{R}^{n}$.
The function $f: \mathbb{R}^{n} \to \mathbb{R}^{n}$ is supposed
to be locally Lipschitz continuous.
Owing to the local Lipschitzness,
there exists a maximal time interval $I(x_{0}) = [ 0, t_{\max}(x_{0}) )$
and a unique continuously differentiable solution
$\varphi: I(x_{0}) \to \mathbb{R}^{n}$ such that
$\dot{\varphi}(t) = f(\varphi(t))$ and $\varphi(0) = x_{0}$
for all $t \in I(x_{0})$~\cite{taylorSafeBacksteppingControl2022}.

Assume that $C \subset \mathbb{R}^{n}$ is rendered
as the $0$-superlevel set of a continuously differentiable function
$h: \mathbb{R}^{n} \to \mathbb{R}$ as $C = \{ x \in \mathbb{R}^{n} \vert h(x) \geq 0 \}$.
The set $C$ is said to be \textit{forward invariant}
if $\varphi(t) \in C$ for all $t \in I(x_{0})$
for any initial condition $x_{0} \in C$.
In this case, the system \eqref{eq:autonomous_system} is said to be
\textit{safe} with respect to the \textit{safe set} $C$~\cite{amesControlBarrierFunctions2019,taylorSafeBacksteppingControl2022}.

\subsection{Control barrier function and safe backstepping}
\label{sec:cbf_and_safe_backstepping}
\subsubsection{Barrier function}
Given dynamical system \eqref{eq:autonomous_system},
barrier function can be used as a tool for safety verification, defined as follows.
\begin{definition}[Barrier function (BF) \cite{amesControlBarrierFunctions2019,taylorSafeBacksteppingControl2022}]
  Let $C \subset \mathbb{R}^{n}$ be the $0$-superlevel set of
  a continuously differentiable function $h: \mathbb{R}^{n} \to \mathbb{R}$
  with $\frac{\partial h}{\partial x} (x) \neq 0$ when $h(x) = 0$.
  The function $h$ is a barrier function for \eqref{eq:autonomous_system} on $C$
  if there exists $\alpha \in \mathcal{K}_{\infty}^{e}$ such that for all $x \in \mathbb{R}^{n}$,
  $\dot{h}(x) \geq - \alpha(h(x))$.
\end{definition}

Consider a nonlinear input-affine control system,
\begin{equation}
  \label{eq:input_affine_system}
  \dot{x} = f(x) + g(x)u,
\end{equation}
and the functions $f: \mathbb{R}^{n} \to \mathbb{R}^{n}$
and $g: \mathbb{R}^{n} \to \mathbb{R}^{n \times m}$ are assumed to be locally Lischitz continuous.
Control barrier functions act a role of synthesizing safe controllers.

\begin{definition}[Control barrier function (CBF) \cite{amesControlBarrierFunctions2019,taylorSafeBacksteppingControl2022}]
  Let $C \subset \mathbb{R}^{n}$ be the $0$-superlevel set of a continuously differentiable
  function $h: \mathbb{R}^{n} \to \mathbb{R}$ with $\frac{\partial h}{\partial x}(x) \neq 0$
  when $h(x) = 0$.
  The function $h$ is a control barrier function for \eqref{eq:input_affine_system} on $C$
  if there exists $\alpha \in \mathcal{K}_{\infty}^{e}$ such that for all $x \in \mathbb{R}^{n}$,
  $\sup_{u \in \mathbb{R}^{m}} \dot{h}(x, u) > -\alpha(h(x))$.
\end{definition}
% The strict inequality is for
% local Lipschitzness of the optimization-based controllers
% and the non-emptiness of the set
% $K_{\text{CBF}}(x) = \{ u \in \mathbb{R}^{m} \vert \dot{h}(x, u) \geq - \alpha(h(x))\}$
% for all $x$~\cite{taylorSafeBacksteppingControl2022,jankovicRobustControlBarrier2018}.
If $h$ is a CBF for \eqref{eq:input_affine_system},
then the set $K_{\text{CBF}}(x) = \{ u \in \mathbb{R}^{m} \vert \dot{h}(x, u) \geq - \alpha(h(x))\}$ is non-empty for all $x$,
and for any locally Lipschitz continuous controller $k$ such that $k(x) \in K_{\text{CBF}}(x)$ for all $x$,
the function $h$ is a BF for the closed-loop system on $C$~\cite{amesControlBarrierFunctions2019,taylorSafeBacksteppingControl2022}.
% It should be pointed out that
% the resulting constraint is affine in $u$.
% Typically,
% the affine inequalities for safety consideration
% are incorporated with quadratic costs to form a quadratic program (QP),
% which yields a locally Lipschitz continuous safe controller\cite{jankovicRobustControlBarrier2018}.

\subsubsection{High-order control barrier function}
High-order CBF (HOCBF) is used for relative degree $r > 1$.
Given $r$-times continuously differentiable function $h: \mathbb{R}^{n} \to \mathbb{R}$,
it is assumed that the control input explicitly appears in $h^{(r)}$
and does not explicitly appear in $h^{(i)}$ for all $i \in \{0, \ldots, r-1\}$.
One can recursively define functions $h_{i}$ with $h_{0} := h$ as follows,
\begin{align}
    h_{i}(x) &:= \dot{h}_{i-1}(x) + \alpha_{i}(h_{i-1}(x)), \forall i \in \{1, \ldots, r-1\}
    \\
    \label{eq:hocbf_constraint}
    h_{r}(x, u) &:= \dot{h}_{r-1}(x) + \alpha_{r}(h_{r-1}(x)) \geq 0,
\end{align}
with $\alpha_{i} \in \mathcal{K}_{\infty}^{e}$
for $i \in \{1, \ldots, r\}$.
Let us define sets $C_{i} := \{x \in \mathbb{R}^{n} \vert h_{i}(x) \geq 0 \}$ for all $i \in \{0, \ldots, r-1\}$.
If $\sup_{u \in \mathbb{R}^{m}} h_{r}(x, u) > 0$ for all $x \in \cap_{i=0}^{r-1} C_{i}$,
$h$ is said to be a HOCBF.
The control input $u$ satisfying the inequality \eqref{eq:hocbf_constraint} implies the forward invariance of $\cap_{i=0}^{r-1} C_{i}$
if the initial state $x_{0}$ is in $\cap_{i=0}^{r-1} C_{i}$~\cite{xiaoHighOrderControlBarrier2022}.

\subsubsection{Safe backstepping}
\label{sec:safe_backstepping}
Safe backstepping, also known as CBF backstepping, is a backstepping method proposed for CBF,
inspired by the Lyapunov backstepping method~\cite{taylorSafeBacksteppingControl2022}.
The safe backstepping provides a tool for synthesizing safe controllers
under the challenge of mixed relative degree.

Consider a nonlinear system in strict feedback form,
\begin{equation}
  \label{eq:strict_feedback_form}
  \begin{split}
    \dot{\xi}_{0} &= f_{0}(z_{0}) + g_{0, \xi}(z_{0}) \xi_{1} + g_{0, u}(z_{0})u_{0},
    \\
    \vdots
    \\
    \dot{\xi}_{l-1} &= f_{l-1}(z_{l-1}) + g_{l-1, \xi}(z_{l-1}) \xi_{l} + g_{l-1, u}(z_{l-1})u_{l-1},
    \\
    \dot{\xi}_{l} &= f_{l}(z_{l}) + g_{l, u}(z_{l})u_{l},
  \end{split}
\end{equation}
with states $\xi_{i} \in \mathbb{R}^{n_{i}}$ and inputs $u_{i} \in \mathbb{R}^{m_{i}}$,
where $z_{i} := [\xi_{0}^{\intercal}, \ldots, \xi_{i}^{\intercal}]^{\intercal}$
for $i \in \{0, \ldots, l\}$.
It is assumed that the functions $f_{i}$, $g_{i, u}$ for $i \in \{0, \ldots, l\}$ are sufficiently smooth.
Also, the functions $g_{i}$ such that $g_{i}(z_{i}) = [g_{i, \xi}(z_{i}), g_{i, u}(z_{i})]$ are assumed to be pseudo-invertible.

Given a continuously differentiable function $h: \mathbb{R}^{n_{0}} \to \mathbb{R}$,
suppose that there exist functions $k_{0, \xi}$ and $k_{0, u}$ such that
\begin{equation}
  \label{eq:safe_backstepping_assumption}
  \begin{split}
    \frac{\partial h_{0}}{\partial \xi_{0}}(z_{0})
    (
    f_{0}(z_{0})
    &+ g_{0, \xi}(z_{0}) k_{0, \xi}(z_{0})
    \\
    &+ g_{0, u}(z_{0}) k_{0, u}(z_{0})
    )
    > - \alpha_{0}(h_{0}(z_{0})),
  \end{split}
\end{equation}
for all $z_{0} \in \mathbb{R}^{n_{0}}$.
That is, the functions $k_{0, \xi}$ and $k_{0, u}$ are safe controllers when regarding $\xi_{1}$ as a control input.
Then,
the controllers $k_{1, \xi}$ and $k_{1, u}$ are defined as
\begin{equation}
  \begin{split}
    \begin{bmatrix}
      k_{1, \xi}(z_{1})
      \\
      k_{1, u}(z_{1})
    \end{bmatrix}
    &= g_{1}(z_{1})^{\dagger}
    \biggl(
    -f_{1}(z_{1})
    + \mu_{0} \left(
      \frac{\partial h_{0}}{\partial \xi_{0}} (z_{0}) g_{0, \xi}(z_{0})
    \right)^{\intercal}
    \\
    &+ \dot{k}_{0, \xi} \vert_{u_{0} = k_{0, u}}
    - \frac{\lambda_{1}}{2} (\xi_{1} - k_{0, \xi}(z_{0}))
    \biggr),
  \end{split}
\end{equation}
where $\mu_{0} \in \mathbb{R}_{>0}$ and $\lambda_{1} \in \mathbb{R}_{>0}$ are positive constants,
and $(\cdot)^{\dagger}$ denotes pseudo-inverse operator.
The controllers $k_{i, \xi}$ and $k_{i, u}$ are recursively designed for $i \in \{2, \ldots, l\}$ as
\begin{equation}
  \begin{split}
    &k_{i}(z_{i})
    = g_{i}(z_{i})^{\dagger}
    \biggl(
    -f_{i}(z_{i})
    \\
    &- \frac{\mu_{i}}{\mu_{i-1}} 
    g_{i-1, \xi}(z_{i-1})^{\intercal} (\xi_{i-1} - k_{i-2, \xi}(z_{i-2}))
    \\
    &+ \dot{k}_{i-1, \xi} \vert_{u_{j} = k_{j, u}, \forall j \in \{0, \ldots, i-1\}}
    - \frac{\lambda_{1}}{2} (\xi_{1} - k_{0, \xi}(z_{0}))
    \biggr),
  \end{split}
\end{equation}
where
$k_{i}(z_{i})
:= [k_{i, \xi}(z_{i})^{\intercal}, k_{i, u}(z_{i})^{\intercal}]^{\intercal}$
for $i \in \{2, \ldots, l-1\}$,
and $k_{l}(z_{l}) := k_{l, u}(z_{l})$.
Let us define a smooth function $h: \mathbb{R}^{n} \to \mathbb{R}$,
\begin{equation}
  \label{eq:safe_backstepping_cbf}
  h(z_{l}) := h_{0}(z_{0}) - \sum_{i=1}^{l} \frac{1}{2 \mu_{i}} \lVert \xi_{i} - k_{i-1, \xi}(z_{i-1}) \rVert^{2},
\end{equation}
with $\mu_{i} \in \mathbb{R}_{>0}$ for $i \in \{1, \ldots, l\}$.
The corresponding set $C \subset \mathbb{R}^{n}$ is defined as follows,
\begin{equation}
  \label{eq:safe_backstepping_set}
  C = \{z_{l} \in \mathbb{R}^{n} \vert h(z_{l}) \geq 0 \},
\end{equation}
where $n = \sum_{i=0}^{l} n_{i}$.
% Note from \eqref{eq:safe_backstepping_cbf} that
% $C \subset C_{0} \times \mathbb{R}^{n_{1}} \times \cdots \times \mathbb{R}^{n_{l}}$.
The following theorem describes the safe backstepping method.

\begin{theorem}[Safe backstepping~\cite{taylorSafeBacksteppingControl2022}]
  \label{thm:safe_backstepping}
  Let $C_{0}$ be the $0$-superlevel set of
  a smooth function $h_{0}: \mathbb{R}^{n_{0}} \to \mathbb{R}$
  with $\frac{\partial h_{0}}{\partial \xi_{0}} \neq 0$
  when $h_{0}(z_{0}) = 0$.
  If there exist smooth functions
  $k_{0, \xi}: \mathbb{R}^{n_{0}} \to \mathbb{R}^{n_{1}}$
  and $k_{0, u}: \mathbb{R}^{n_{0}} \to \mathbb{R}^{m_{0}}$
  and a globally Lipschitz continuous function $\alpha_{0} \in \mathcal{K}_{\infty}^{e}$
  such that \eqref{eq:safe_backstepping_assumption} holds,
  then the function $h$ defined in \eqref{eq:safe_backstepping_cbf}
  is a control barrier function for the system \eqref{eq:strict_feedback_form}
  on the set $C$ defined in \eqref{eq:safe_backstepping_set}.
\end{theorem}

\subsection{Multicopter dynamics}
\label{sec:multicopter_dynamics}
The dynamics of multicopter can be expressed as \cite{leeGeometricTrackingControl2010},
\begin{align}
  \dot{p} &= v,
  \\
  \label{eq:v_dot}
  \dot{v} &= g e_{3} - \frac{1}{m} T R e_{3},
  \\
  \label{eq:R_dot}
  \dot{R} &= R \omega^{\times},  % BE CAREFUL FOR THE DEFINITION OF R
  \\
  \label{eq:omega_dot}
  \dot{\omega} &= J^{-1} (M -\omega \times J \omega),
\end{align}
where
$p := [p_{x}, p_{y}, p_{z}]^{\intercal} \in \mathbb{R}^{3}$
and $v := [v_{x}, v_{y}, v_{z}]^{\intercal} \in \mathbb{R}^{3}$ are the position and velocity in the inertial frame,
respectively,
$g \in \mathbb{R}_{>0}$ is the magnitude of gravitational acceleration,
$e_{3} = [0, 0, 1]$ is the $z$-axis unit vector in inertial frame,
$m \in \mathbb{R}_{>0}$ is the mass of multicopter,
$T \in \mathbb{R}$ is the total thrust,
$R \in \textrm{SO}(3)$ is the rotation matrix from inertial to body-fixed frames,
and $(\cdot)^{\times}: \mathbb{R}^{3} \to \text{so}(3)$ is the cross-product-map,
such that $x^{\times}y = x \times y, \forall x, y \in \mathbb{R}^{3}$.
$\omega := [\omega_{x}, \omega_{y}, \omega_{z}]^{\intercal} \in \mathbb{R}^{3}$ is the angular velocity in the body-fixed frame,
$J \in \mathbb{R}^{3 \times 3}$ is the moment of inertia,
and $M \in \mathbb{R}^{3}$ is the applied torque.

The control allocation of multicopter can be written as,
\begin{equation}
  \nu := [T, M^{\intercal}]^{\intercal} = B u,
\end{equation}
where $u \in \mathbb{R}^{m_{\text{rtr}}}$ is the rotor thrust vector,
and $B \in \mathbb{R}^{4 \times m_{\text{rtr}}}$ is the control effectiveness matrix depending on the multicopter configuration,
e.g., quadcopter, hexacopter, etc.
Then, the state is
$x := [p^{\intercal}, v^{\intercal}, \text{vec}(R)^{\intercal}, \omega^{\intercal}]^{\intercal} \in \mathbb{R}^{18}$
with control input $\nu \in \mathbb{R}^{4}$,
where $\text{vec}(\cdot)$ is the vectorize operator.
\autoref{fig:hexacopter_illustration} illustrates a hexacopter ($m_{\text{rtr}} = 6$).
\begin{figure}[!h]
    \centering
    \begin{subfigure}[b]{0.24\textwidth}
        \centering
        \includegraphics[width=\linewidth]{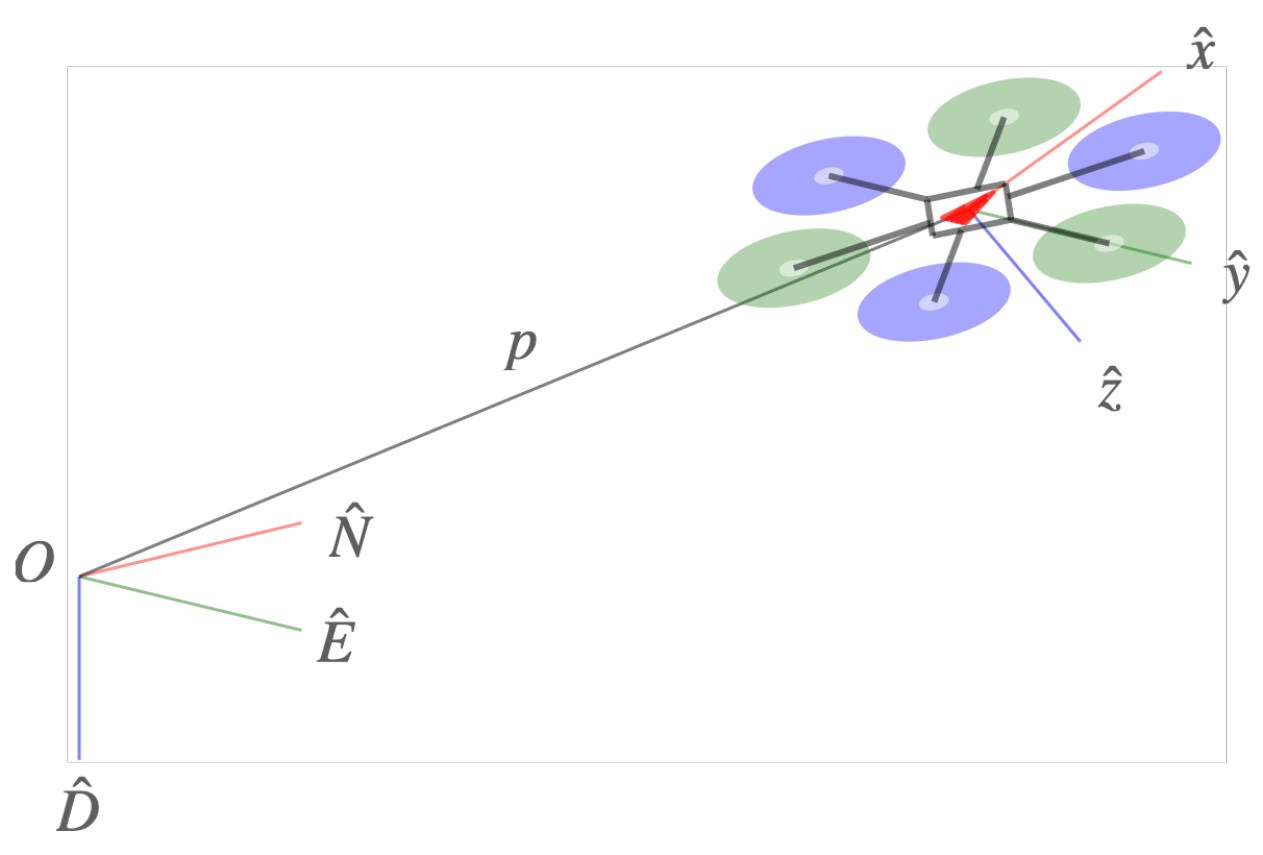}
        \caption{}
    \end{subfigure}
    \begin{subfigure}[b]{0.20\textwidth}
        \centering
        \includegraphics[width=\linewidth]{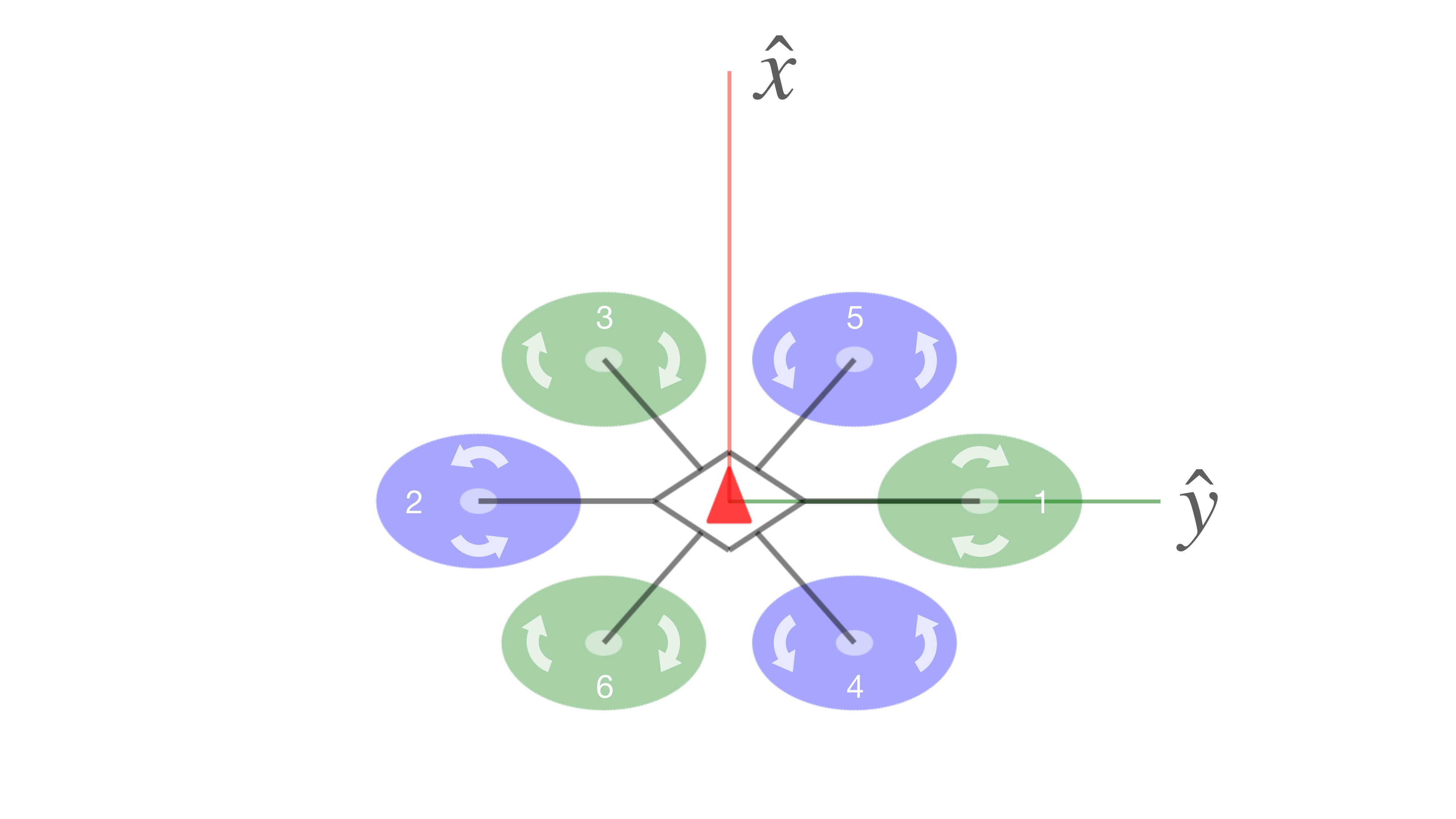}
        \caption{}
    \end{subfigure}
    \caption{
      Illustration of multicopter. (a) Multicopter with six rotors (hexacopter-x) and frames, (b) top view
    }
    \label{fig:hexacopter_illustration}
\end{figure}

\section{Main Results}
\label{sec:main_results}
\subsection{Multicopter dynamics reformulation}
\label{sec:reformulation}
The main challenge to utilize the CBF methods for multicopter is that
the safe backstepping is not directly applicable to the multicopter dynamics in \autoref{sec:multicopter_dynamics},
which suffers from the mixed-relative-degree challenge and is not strict feedback form.
For example, the term $-\frac{1}{m} T R e_{3}$ in \eqref{eq:v_dot}
is formed as the multiplication of
a control input $T$ and an outer-loop state $z_{B} := R e_{3} \in \mathbb{S}^{2}$,
representing the $z$-axis of the body-fixed frame.

To deal with this issue,
let us reformulate the multicopter dynamics
based on Lyapunov backstepping multicopter controller~\cite{falconiAdaptiveFaultTolerant2016}.
First, the force vector $f \in \mathbb{R}^{3}$ is expressed as $f = -T z_{B}$.
By augmenting the total thrust $T$ as a state variable
and its time derivative $\dot{T}$ as an input variable,
\eqref{eq:v_dot} can be reformulated as

\begin{equation}
  \label{eq:f_dot}
  \dot{v} = g e_{3} + \frac{1}{m} f,
  \quad
  \dot{f} = - T R A \omega_{xy} -\dot{T} z_{B},
\end{equation}
where $A = [0, 1; -1, 0; 0, 0] \in \mathbb{R}^{3 \times 2}$ is an auxiliary matrix,
and $\omega_{xy} := [\omega_{x}, \omega_{y}]^{\intercal} \in \mathbb{R}^{2}$.
Augmented state and control input are defined as
$x_{a} := [x^{\intercal}, T]^{\intercal} \in \mathbb{R}^{n_{a}}$
and $\nu_{a} := [\dot{T}, M^{\intercal}]^{\intercal} \in \mathbb{R}^{4}$,
respectively.

\subsection{Safety in rotational dynamics}
In this study,
angular speed and total thrust direction limits
are considered as the safety constraints in rotational dynamics~\cite{kimSafeAttitudeController2023}.

\subsubsection{Angular velocity safety}
Consider the following function, $h_{\omega}(x) = 1 - \omega^{\intercal} P_{\omega} \omega$,
where $P_{\omega} \in \mathbb{R}^{3 \times 3}$ is a positive semi-definite diagonal matrix for angular velocity constraint.
For example, $h_{\omega} \geq 0 \Longleftrightarrow \lVert \omega \rVert \leq \overline{\omega}$
if $P_{\omega} = \text{diag}(1/\overline{\omega}^{2}, 1/\overline{\omega}^{2}, 1/\overline{\omega}^{2})$
for given maximum angular speed $\overline{\omega} \in \mathbb{R}_{>0}$.

Because the relative degree of $h_{\omega}$ is one,
the angular velocity safety is guaranteed by the conventional CBF method
with the control input satisfying the following inequality,
\begin{equation}
  \label{eq:angular_velocity_safety_constraint}
  \dot{h}_{\omega}(x, M) + \alpha_{\omega}(h_{\omega}(x)) \geq 0,
\end{equation}
with $\alpha_{\omega} \in \mathcal{K}_{\infty}^{e}$.

\subsubsection{Total thrust direction safety}
Given unit vector $z_{B_{d}} \in \mathbb{S}^{2}$ indicating the desired $z$-axis in the inertial frame,
consider the following function, $h_{0, z_{B}}(x) = z_{B}^{\intercal} z_{B_{d}} - \cos(\overline{\theta}_{z_{B}})$,
where $\overline{\theta}_{z_{B}} \in [0, \pi]$ is the maximum deviation angle of body-fixed $z$-axis.
Note that $h_{0, z_{B}} \geq 0 \Longleftrightarrow \theta_{z_{B}} \leq \overline{\theta}_{z_{B}}$,
where $\theta_{z_{B}} := \arccos(z_{B}^{\intercal} z_{B_{d}}) \in [0, \pi]$.

Because the relative degree of $h_{0, z_{B}}$ is two,
the total thrust direction safety can be guaranteed by the high-order CBF method
with the control input satisfying the following inequality,
\begin{equation}
  \label{eq:total_thrust_direction_safety_constraint}
  \begin{split}
    h_{1, z_{B}}(x) &:= \dot{h}_{0, z_{B}}(x) + \alpha_{1, z_{B}}(h_{0, z_{B}}(x)),
    \\
    h_{2, z_{B}}(x, M) &:= \dot{h}_{1, z_{B}}(x, M) + \alpha_{2, z_{B}}(h_{1, z_{B}}(x)) \geq 0,
  \end{split}
\end{equation}
with functions $\alpha_{1, z_{B}}, \alpha_{2, z_{B}} \in \mathcal{K}_{\infty}^{e}$.

\subsection{Safety in translational dynamics}
In this study,
velocity and position limits are considered as the safety constraints in translational dynamics.

\subsubsection{Velocity safety}
Consider the following function, $h_{0, v}(x) = 1 - v^{\intercal} P_{v} v$,
where $P_{v} \in \mathbb{R}^{3 \times 3}$ is a positive semi-definite digonal matrix for velocity constraint.
For example, $h_{0, v} \geq 0 \Longleftrightarrow \lVert v \rVert \leq \overline{v}$
if $P_{v} = \text{diag}(1/\overline{v}^{2}, 1/\overline{v}^{2}, 1/\overline{v}^{2})$
for given maximum speed $\overline{v} \in \mathbb{R}_{>0}$.

Because the dynamics related to $h_{0, v}$ are in strict feedback form with the reformulation,
the velocity safety is guaranteed by the safe backstepping with the control input satisfying the following inequality,
\begin{equation}
  \begin{split}
  h_{v}(x_{a})
  &:= h_{0, v}(x)
  \\
  &- \sum_{i=1}^{2} \frac{1}{2 \mu_{i, v}} \lVert \xi_{i, v} - k_{i-1, \xi_{v}}(z_{i-1, v}, t) \rVert^{2},
  \\
  \label{eq:velocity_safety_constraint}
  &\dot{h}_{v}(x_{a}, \nu_{a}) + \alpha_{0, v}(h_{v}(x_{a}))
  \geq 0,
  \end{split}
\end{equation}
for $\mu_{i, v} \in \mathbb{R}_{>0}$
where $\xi_{0, v} := v$, $\xi_{1, v} := f$, $\xi_{2, v} := \omega_{xy}$,
$u_{1, v} := \dot{T}$,
$u_{2, v} := M$,
and the controllers $k_{i-1, \xi_{v}}$ are designed by the safe backstepping procedure for $i \in \{1, 2\}$.

\subsubsection{Position safety}
Consider the following function, $h_{0, p}(x) = 1 - (p - p_{d})^{\intercal} P_{p} (p - p_{d})$,
where $p_{d} \in \mathbb{R}^{3}$ is the center of the position safety set,
and $P_{p} \in \mathbb{R}^{3 \times 3}$ is a positive semi-definite weight matrix for position constraint.
For example, $h_{0, p} \geq 0 \Longleftrightarrow \lVert p - p_{d} \rVert \leq \overline{p}$
if $P_{p} = \text{diag}(1/\overline{p}^{2}, 1/\overline{p}^{2}, 1/\overline{p}^{2})$
for given maximum displacement $\overline{p} \in \mathbb{R}_{>0}$.

Similar to the case of velocity safety,
the position safety is guaranteed by the safe backstepping with the control input satisfying the following inequality,
\begin{equation}
  \label{eq:position_safety_constraint}
  \begin{split}
  h_{p}(x_{a})
  &:= h_{0, p}(x)
  \\
  &- \sum_{i=1}^{3} \frac{1}{2 \mu_{i, p}} \lVert \xi_{i, p} - k_{i-1, \xi_{p}}(z_{i-1, p}, t) \rVert^{2},
  \\
  &\dot{h}_{p}(x_{a}, \nu_{a}) + \alpha_{0, p}(h_{p}(x_{a})) \geq 0,
  \end{split}
\end{equation}
for $\mu_{i, p} \in \mathbb{R}_{>0}$
where $\xi_{0, p} := p$, $\xi_{1, p} := v$, $\xi_{2, p} := f$, $\xi_{3, p} := \omega_{xy}$,
$u_{2, p} := \dot{T}$,
$u_{3, p} := M$,
and the controllers $k_{i-1, \xi_{p}}$ are designed by the safe backstepping procedure for $i \in \{1, 2, 3\}$.
\begin{remark}
  The arguments of the designed controllers for safety with respect to velocity and position,
  $k_{i-1, \xi_{v}}$ in \eqref{eq:velocity_safety_constraint} and $k_{i-1, \xi_{p}}$ in \eqref{eq:position_safety_constraint},
  have $t$ because $g_{1, \xi, v} = -T R A$ and $g_{1, u, v} = -z_{B}$
  are not explicitly expressed by $z_{1, v} = [\xi_{0, v}^{\intercal}, \xi_{1, v}^{\intercal}]^{\intercal}$.
\end{remark}

\subsection{Proposed safe controller}
For the safe backstepping in safety with respect to velocity and position,
virtual controllers $k_{0, \xi_{v}}$ and $k_{0, \xi_{p}}$
should be designed first.
Consider the proposed virtual controllers
$k_{0, \xi_{v}}$ and $k_{0, \xi_{p}}$
and functions $\alpha_{0, v} \in \mathcal{K}_{\infty}^{e} $ and $\alpha_{0, p} \in \mathcal{K}_{\infty}^{e} $
for safe backstepping:
\begin{align}
  \label{eq:velocity_virtual_controller}
  k_{0, \xi_{v}}(z_{0, v}, t) &= -m
  \left(
    g e_{3} + \frac{1}{2} c_{v} v
  \right).
  \\
  \label{eq:position_virtual_controller}
  k_{0, \xi_{p}}(z_{0, p}, t) &= -\frac{1}{2} c_{p} \left(p - p_{d}\right),
  \\
  \label{eq:velocity_alpha}
  \alpha_{0, v}(r) &= a_{0, v} r
  \\
  \alpha_{0, p}(r) &= a_{0, p} r
\end{align}
for $a_{0_{v}}, a_{0_{p}}, c_{v}, c_{p} \in \mathbb{R}_{>0}$
such that $c_{i} \geq a_{0, i}$ for $i \in \{v, p\}$.

An optimization-based safety-critical controller $k$ for multicopter is proposed as follows,
\begin{equation}
  \label{eq:proposed_controller}
  \begin{split}
    k(x_{a}) &= \argmin_{\nu_{a} \in \mathbb{R}^{4}} \frac{1}{2}
    \lVert
    \nu_{a} - k_{d}(x_{a})
    \rVert^{2}
    \\
    \text{s.t. } &
    \eqref{eq:angular_velocity_safety_constraint},
    \eqref{eq:total_thrust_direction_safety_constraint},
    \eqref{eq:velocity_safety_constraint},
    \eqref{eq:position_safety_constraint},
  \end{split}
\end{equation}
where $k_{d}$ is a nominal controller.
The following theorem shows the main result of this study.

\begin{theorem}[Main result]
  Given a nominal controller $k_{d}$,
  suppose that the initial augmented state $x_{a}(0)$ is in $\cap_{i \in \{\omega, z_{B}, v, p\}} C_{i}$,
  where the sets are defined as follows,
  \begin{align}
    C_{\omega} &:= \{x_{a} \in \mathbb{R}^{n_{a}} \vert h_{\omega}(x) \geq 0 \},
    \\
    C_{z_{B}} &:=\cap_{i \in \{0, 1\}} \{x_{a} \in \mathbb{R}^{n_{a}} \vert h_{i, z_{B}}(x) \geq 0 \},
    \\
    C_{v} &:= \{x_{a} \in \mathbb{R}^{n_{a}} \vert h_{v}(x_a) \geq 0 \},
    \\
    C_{p} &:= \{x_{a} \in \mathbb{R}^{n_{a}} \vert h_{p}(x_a) \geq 0 \}.
  \end{align}
  Then,
  the state $x(t)$ remains in the safety constraints for all $t \in \mathbb{R}_{\geq 0}$,
  that is,
  $h_{\omega}(x(t)) \geq 0$,
  $h_{0, z_{B}}(x(t)) \geq 0 $,
  $h_{0, v}(x(t)) \geq 0 $,
  $h_{0, p}(x(t)) \geq 0$,
  $\forall t \in \mathbb{R}_{\geq 0}$,
  under the proposed controller in \eqref{eq:proposed_controller}
  with virtual controllers \eqref{eq:velocity_virtual_controller} and \eqref{eq:position_virtual_controller}.
  Also, the optimization problem in \eqref{eq:proposed_controller} is a QP with affine inequality constraints.
\end{theorem}
\begin{proof}
  It is straightforward to show the safety satisfaction for angular velocity and total thrust direction
  based on the conventional CBF and HOCBF introduced in \autoref{sec:cbf_and_safe_backstepping},
  and the safe backstepping procedures for safety with respect to velocity and position are similar to each other.
  Therefore, only the velocity safety will be discussed in the proof.
  Let us show that the virtual controller in \eqref{eq:velocity_virtual_controller}
  satisfies \eqref{eq:safe_backstepping_assumption}.
  Taking $\alpha_{0, v}$ in \eqref{eq:velocity_alpha} implies
  \begin{equation}
    \begin{split}
    &\frac{\partial h_{0, v}}{\partial \xi_{0, v}}(z_{0, v}) (
      f_{0, v}(z_{0, v}) + g_{0, \xi_{v}}(z_{0, v}) k_{0, \xi_{v}} (z_{0, v})
    )
    \\
    &+ \alpha_{0, v}(h_{0, v}(z_{0, v}))
    = a_{0, v} + (c_{v} - a_{0, v}) v^{\intercal} P_{v} v > 0,
    \end{split}
  \end{equation}
  and the velocity safety is guaranteed by \autoref{thm:safe_backstepping}.

  Next, let us show that the resulting optimization problem is a QP with affine inequality constraints.
  Similar to the previous part,
  it is straightforward to show that
  \eqref{eq:angular_velocity_safety_constraint} and \eqref{eq:total_thrust_direction_safety_constraint}
  are affine in $M$.
  Also, the procedures for safety with respect to velocity and position are similar to each other.
  Again, only the velocity safety will be discussed.
  The main concern comes from $\frac{\partial k_{1, \xi_{v}}}{\partial t} (z_{1, v}, t)$
  with the fact that $\frac{\partial k_{0, \xi_{v}}}{\partial t} (z_{0, v}, t) \equiv 0$
  because the argument $t$ does not explicitly appear in the right-hand sides
  of \eqref{eq:velocity_virtual_controller} and \eqref{eq:position_virtual_controller}.
  After laborious calculations,
  we have
  \begin{equation}
    \label{eq:laborious_calculation}
    \begin{split}
      &\frac{\partial }{\partial t} \begin{bmatrix}
        k_{1, \xi_{v}}(z_{1, v}, t)
        \\
        k_{1, u_{v}}(z_{1, v}, t)
      \end{bmatrix}
      = \frac{d g_{1, v}^{-1}}{dt} \biggl(
        - f_{1, v}
        + \mu_{1, v} \left(\frac{\partial h_{0, v}}{\partial \xi_{0, v}} g_{0, \xi_{v}}\right)^{\intercal}
        \\
      &+ \frac{\partial k_{0, \xi_{v}}}{\partial \xi_{0, v}} \left(
          f_{0, v} + g_{0, \xi} \xi_{1, v}
        \right)
        - \frac{\lambda_{1, v}}{2} \left(\xi_{1, v} - k_{0, \xi_{v}}\right)
      \biggr).
    \end{split}
  \end{equation}
  Note that function arguments are omitted in the above equation for better readability.
  In the right-hand side of \eqref{eq:laborious_calculation},
  the terms except for $\frac{d g_{1, v}^{-1}}{d t}$ does not contain any control inputs.
  Because $\frac{d g_{1, v}^{-1}}{d t} = -g_{1, v}^{-1} \dot{g}_{1, v} g_{1, v}^{-1} $
  and $\dot{g}_{1, v} = -\left[\left(\dot{T}R + TR \omega^{\times}\right)A, RA \omega_{xy}\right] \in \mathbb{R}^{3 \times 3}$,
  it can be shown that the control input $\dot{T}$ is affine in the inequality.
  Note also that the terms affine in $M$ comes from the rest by following the original safe backstepping
  in \autoref{thm:safe_backstepping}.
  Therefore, the resulting inequality is affine in $\nu_{a}$.
\end{proof}
It should be pointed out that
\textbf{the affine inequality result is not trivial}.
This supports that the reformulation makes it possible
not just to incorporate the safe backstepping but also to incorporate
state-of-the-art efficient and reliable convex optimization solvers
as well as the previous analyses on safety control with QP formulation,
for example, Lipschitzness of the controller under regularity conditions~\cite{jankovicRobustControlBarrier2018}.

% \begin{remark}
%   Users can impose only some parts of the suggested safety constraints or augment more
%   depending on the requirements.
%   For example, if only velocity and position safety constraints are necessary,
%   then the inequality constraints in \eqref{eq:proposed_controller}
%   can be replaced with \eqref{eq:velocity_safety_constraint} and \eqref{eq:position_safety_constraint}.
% \end{remark}

\section{Numerical simulation}
\label{sec:numerical_simulation}
Numerical simulation is performed with a hexacopter model~\cite{kimControlAllocationSwitching2021},
which adopts the specification of a quadcopter model~\cite{leeGeometricTrackingControl2010} and the control allocation for hexacopter model with hexacopter-X configuration~\cite{px4}.
A Lyapunov backstepping position controller is used as the nominal controller in \eqref{eq:proposed_controller}
with yaw rate regulation~\cite{falconiAdaptiveFaultTolerant2016}.
The reference trajectory is given as $p_{\text{ref}}(t) = [2.5 \cos (0.5 t), 2.5 \sin (0.5 t), -5 + 2.5 \sin (0.25 t)]^{\intercal}$.
Pseudo-inverse control allocation is used,
and each rotor thrust is saturated as $u_{i} \in [0, 0.6371 mg]$.
The bounds are based on \cite{achtelikDesignMultiRotor2012}, slightly modified for hexacopter.
For each time instant with time step $\Delta t = 0.005$s,
the optimization problem \eqref{eq:proposed_controller} is solved by ECOS~\cite{domahidiECOSSOCPSolver2013} to obtain the zero-order-hold control input.

The simulation settings are as follows:
$p(0) = [0, 0, -5]^{\intercal}$m,
$v(0) = [0, 1.25, 0.625]^{\intercal}$m/s,
$(\phi, \theta, \psi) = (-15, 15, 90)$deg,
$\omega = [15, 15, 0]^{\intercal}$deg/s,
$\overline{\omega} = 360$deg/s,
$\overline{\theta}_{z_{B}} = 30$deg,
$\overline{v} = 2$m/s,
$\overline{p} = 3$m,
and $p_{d} = [0, 0, -5]^{\intercal}$m,
where
$\phi$, $\theta$, $\psi$ denote roll, pitch, yaw with ZYX rotation, respectively,
and $P_{i} = \text{diag}(1/\overline{i}^{2}, 1/\overline{i}^{2}, 1/\overline{i}^{2})$ for $i \in \{\omega, v, p\}$.
Every class $\mathcal{K}_{\infty}^{e}$ function is set as a linear function.

\autoref{fig:simulation_result}
shows the simulation result of the position controller with and without the proposed method.
The computation time of the optimization problem is $2$ms on AMD Ryzen 9 5900HS CPU.
At the beginning of the response,
the nominal controller shows faster convergence to the reference trajectory
with high speed, total thrust direction deviation, and angular speed.
Because the nominal controller does not consider the position safety,
the distance from $p_{d}$ is larger than $\overline{p}$ for a long time duration,
which indicates that the position is not in the safe set.
On the other hand,
the proposed safe controller mitigates the excessive maneuver in the beginning,
which results in that the safety constraints are all satisfied.
It is also interesting to look at the response between $10$s and $15$s;
the safe controller does not need to interfere
the nominal controller for position safety,
and the response becomes the same as that of nominal controller.
After that, to maintain the safe behaviour,
the proposed method interferes the nominal controller again.
In summary,
the proposed method does not violate any safety constraints considered in this study,
while the nominal controller does.
It should be pointed out that the proposed method does not utilize cascade control system
in which high-level safety constraints may be violated~\cite{khanBarrierFunctionsCascaded2020}.

\begin{figure*}[!ht]
    \centering
    \begin{subfigure}[b]{0.42\textwidth}
        \centering
        \includegraphics[width=\linewidth]{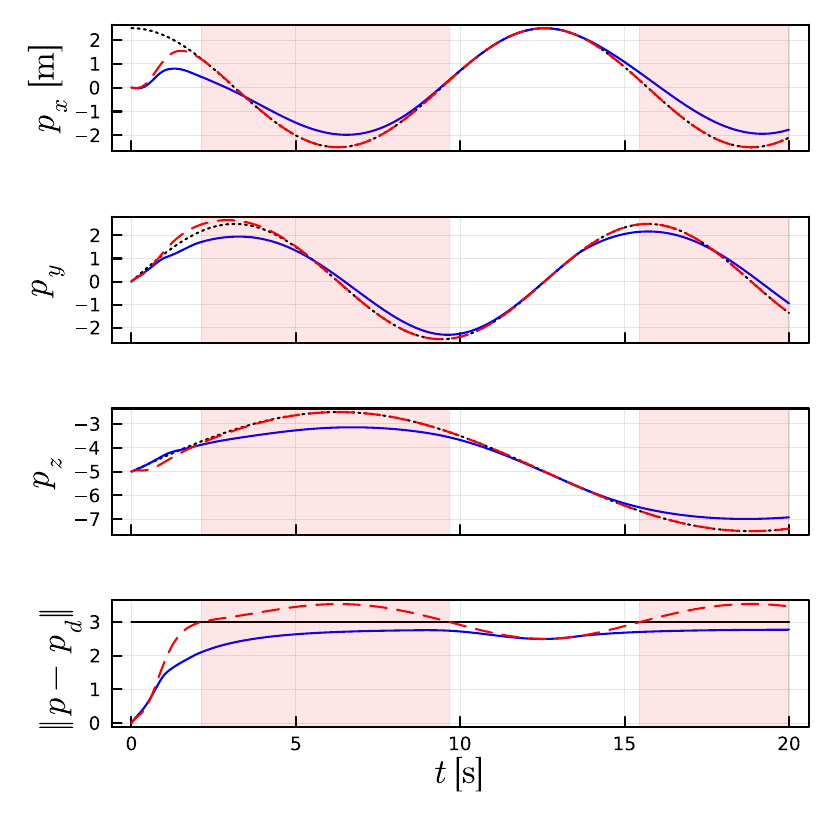}
        \caption{Position}
    \end{subfigure}
    \begin{subfigure}[b]{0.42\textwidth}
        \centering
        \includegraphics[width=\linewidth]{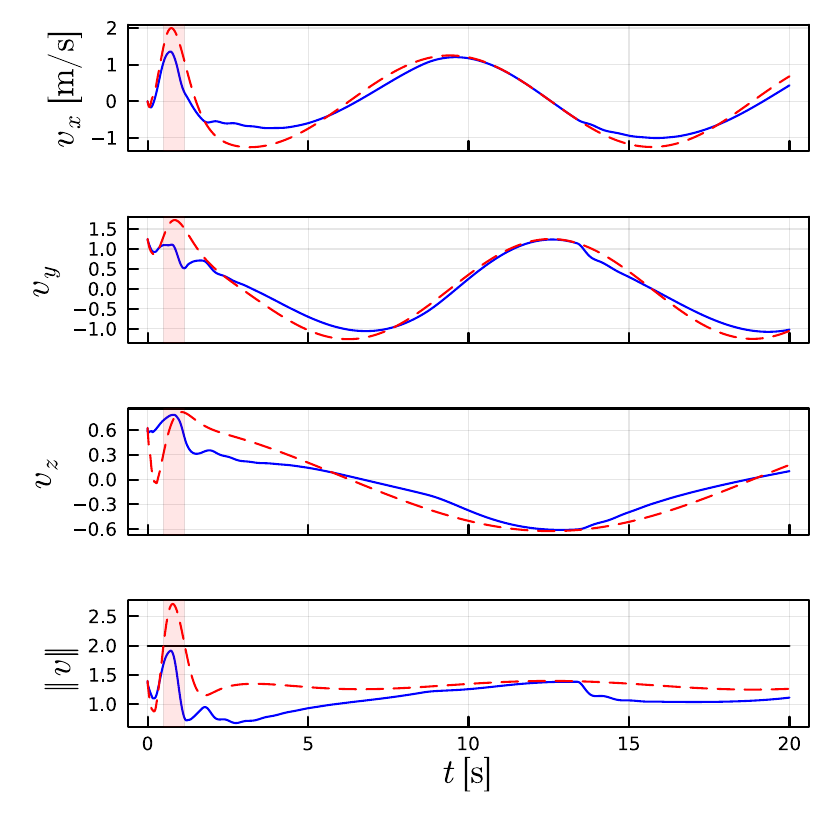}
        \caption{Velocity}
    \end{subfigure}

    \begin{subfigure}[b]{0.42\textwidth}
        \centering
        \includegraphics[width=\linewidth]{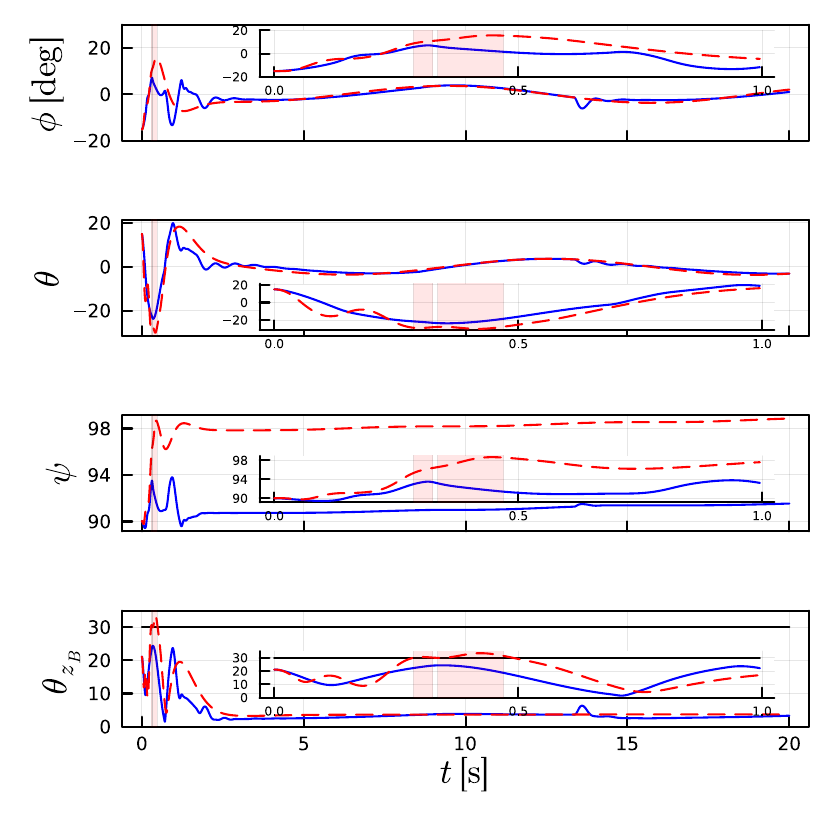}
        \caption{Euler angle}
    \end{subfigure}
    \begin{subfigure}[b]{0.42\textwidth}
        \centering
        \includegraphics[width=\linewidth]{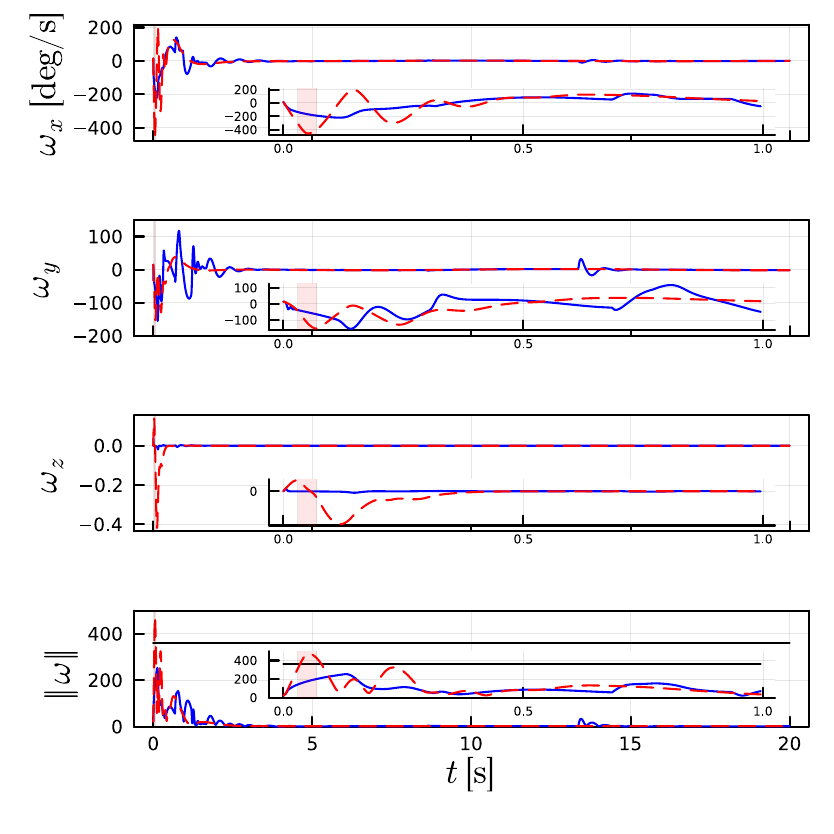}
        \caption{Angular velocity}
    \end{subfigure}

    \begin{subfigure}[b]{0.84\textwidth}
        \centering
        \includegraphics[width=\linewidth]{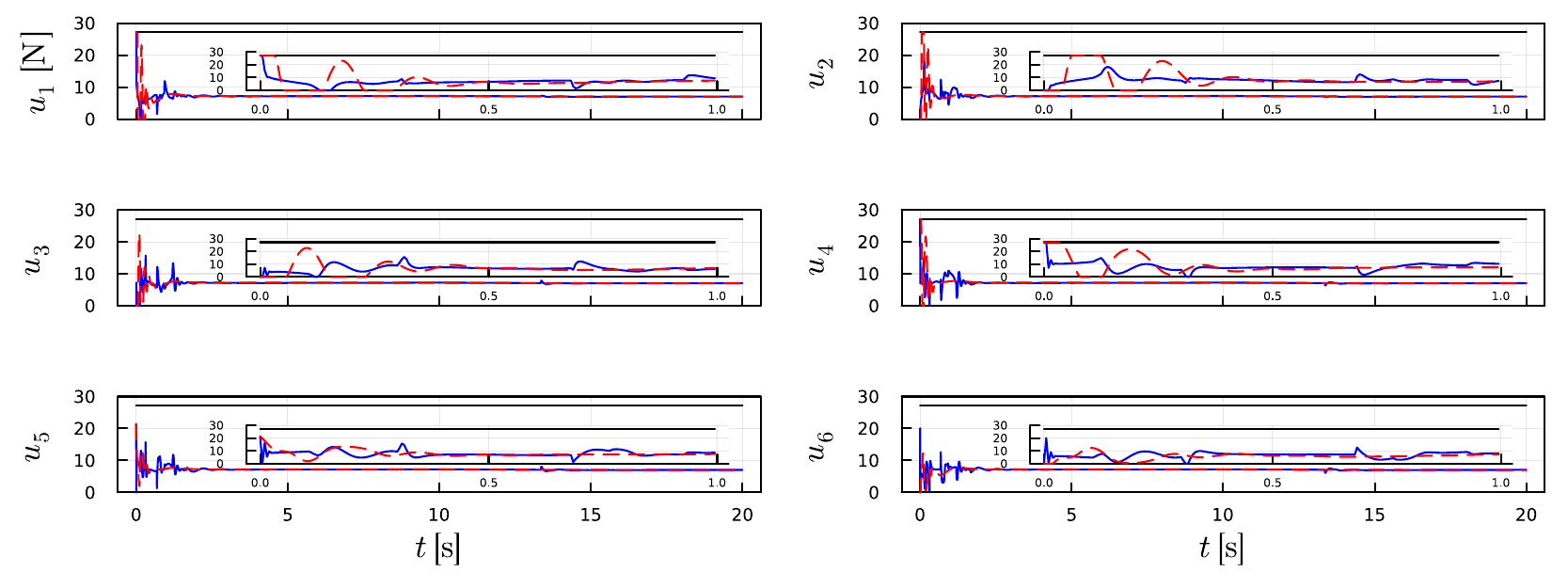}
        \caption{Rotor thrust}
    \end{subfigure}
    \caption{
      Blue solid and red dashed lines denote the trajectories with and without the proposed method,
      respectively.
      Black solid lines denote the maximum values of corresponding safety constraints or applied rotor thrusts.
      Red backgrounds indicate safety violation for nominal controller.
      (a) position and the magnitude of position error (black dotted line: reference trajectory),
      (b) velocity and speed,
      (c) Euler angles and z-axis angle,
      (d) angular velocity and angular speed,
      and (e) rotor thrust.
    }
    \label{fig:simulation_result}
\end{figure*}

\section{Conclusion}
\label{sec:conclusion}
A safe controller was proposed for multicopter considering
various safety constraints with respect to angular velocity, total thrust direction, velocity, and position.
Multicopter dynamics were reformulated to make the dynamics be in strict feedback form,
and conventional and high-order control barrier functions (CBFs) and CBF backstepping
were utilized to satisfy the safety constraints.
The proposed safe controller does not consider a cascade controller design.
Numerical simulation demonstrated that multicopter can track a given reference position trajectory
using the proposed controller
without violating safety constraints considered in this study.

Further studies considering different types of safety constraints and related CBFs
for multicopter as well as different virtual controllers for safe backstepping are required.
Future works include the adaptive safe controller for fault-tolerant control.

\bibliographystyle{IEEEtran}
\bibliography{references}

\begin{thebibliography}{10}
\providecommand{\url}[1]{#1}
\csname url@samestyle\endcsname
\providecommand{\newblock}{\relax}
\providecommand{\bibinfo}[2]{#2}
\providecommand{\BIBentrySTDinterwordspacing}{\spaceskip=0pt\relax}
\providecommand{\BIBentryALTinterwordstretchfactor}{4}
\providecommand{\BIBentryALTinterwordspacing}{\spaceskip=\fontdimen2\font plus
\BIBentryALTinterwordstretchfactor\fontdimen3\font minus
  \fontdimen4\font\relax}
\providecommand{\BIBforeignlanguage}[2]{{%
\expandafter\ifx\csname l@#1\endcsname\relax
\typeout{** WARNING: IEEEtran.bst: No hyphenation pattern has been}%
\typeout{** loaded for the language `#1'. Using the pattern for}%
\typeout{** the default language instead.}%
\else
\language=\csname l@#1\endcsname
\fi
#2}}
\providecommand{\BIBdecl}{\relax}
\BIBdecl

\bibitem{amesControlBarrierFunctions2019}
A.~D. Ames, S.~Coogan, M.~Egerstedt, G.~Notomista, K.~Sreenath, and P.~Tabuada,
  ``\BIBforeignlanguage{en}{Control {Barrier} {Functions}: {Theory} and
  {Applications}},'' in \emph{\BIBforeignlanguage{en}{2019 18th {European}
  {Control} {Conference} ({ECC})}}, Naples, Italy, Jun. 2019.

\bibitem{aubinViabilityTheory2009}
J.-P. Aubin, \emph{\BIBforeignlanguage{en}{Viability {Theory}}}.\hskip 1em plus
  0.5em minus 0.4em\relax Boston, MA: Birkhäuser Boston, 2009.

\bibitem{bemporadHierarchicalHybridModel2009}
A.~Bemporad, C.~Pascucci, and C.~Rocchi, ``\BIBforeignlanguage{en}{Hierarchical
  and {Hybrid} {Model} {Predictive} {Control} of {Quadcopter} {Air}
  {Vehicles}},'' \emph{\BIBforeignlanguage{en}{IFAC Proceedings Volumes}},
  vol.~42, no.~17, pp. 14--19, 2009.

\bibitem{liFinitetimeControlQuadrotor2021}
X.~Li, H.~Zhang, W.~Fan, C.~Wang, and P.~Ma,
  ``\BIBforeignlanguage{en}{Finite-time {Control} for {Quadrotor} based on
  {Composite} {Barrier} {Lyapunov} {Function} with {System} {State}
  {Constraints} and {Actuator} {Faults}},''
  \emph{\BIBforeignlanguage{en}{Aerospace Science and Technology}}, vol. 119,
  Article no. 107063, 2021.

\bibitem{khanBarrierFunctionsCascaded2020}
M.~Khan, M.~Zafar, and A.~Chatterjee, ``\BIBforeignlanguage{en}{Barrier
  {Functions} in {Cascaded} {Controller}: {Safe} {Quadrotor} {Control}},'' in
  \emph{\BIBforeignlanguage{en}{2020 {American} {Control} {Conference}
  ({ACC})}}, Denver, CO, Jul. 2020.

\bibitem{falconiAdaptiveFaultTolerant2016}
G.~P. Falconí and F.~Holzapfel, ``Adaptive {Fault} {Tolerant} {Control}
  {Allocation} for a {Hexacopter} {System},'' in \emph{2016 {American}
  {Control} {Conference} ({ACC})}, Boston, MA, Jul. 2016.

\bibitem{kimSafeAttitudeController2023}
J.~Kim, H.~Lee, and Y.~Kim, ``\BIBforeignlanguage{en}{Safe {Attitude}
  {Controller} {Design} for {Multicopter} via {High}-order {Control} {Barrier}
  {Function}},'' in \emph{\BIBforeignlanguage{en}{Aerospace {Europe}
  {Conference}, {Joint} 10th {EUCASS} and 9th {CEAS}}}, Lausanne, Switzerland,
  Jul. 2023.

\bibitem{xiaoHighOrderControlBarrier2022}
W.~Xiao and C.~Belta, ``\BIBforeignlanguage{en}{High-{Order} {Control}
  {Barrier} {Functions}},'' \emph{\BIBforeignlanguage{en}{IEEE Transactions on
  Automatic Control}}, vol.~67, no.~7, pp. 3655--3662, 2022.

\bibitem{taylorSafeBacksteppingControl2022}
A.~J. Taylor, P.~Ong, T.~G. Molnar, and A.~D. Ames,
  ``\BIBforeignlanguage{en}{Safe {Backstepping} with {Control} {Barrier}
  {Functions}},'' in \emph{\BIBforeignlanguage{en}{2022 {IEEE} 61st
  {Conference} on {Decision} and {Control} ({CDC})}}, Canc\'un, Mexico, Dec.
  2022.

\bibitem{leeGeometricTrackingControl2010}
T.~Lee, M.~Leok, and N.~H. McClamroch, ``\BIBforeignlanguage{en}{Geometric
  {Tracking} {Control} of a {Quadrotor} {UAV} on {SE}(3)},'' in
  \emph{\BIBforeignlanguage{en}{49th {IEEE} {Conference} on {Decision} and
  {Control} ({CDC})}}, Atlanta, GA, Dec. 2010.

\bibitem{jankovicRobustControlBarrier2018}
M.~Jankovic, ``\BIBforeignlanguage{en}{Robust {Control} {Barrier} {Functions}
  for {Constrained} {Stabilization} of {Nonlinear} {Systems}},''
  \emph{\BIBforeignlanguage{en}{Automatica}}, vol.~96, pp. 359--367, 2018.

\bibitem{kimControlAllocationSwitching2021}
J.~Kim, H.~Lee, S.-h. Kim, M.~Kim, and Y.~Kim,
  ``\BIBforeignlanguage{en}{Control {Allocation} {Switching} {Scheme} for
  {Fault} {Tolerant} {Control} of {Hexacopter}},'' in
  \emph{\BIBforeignlanguage{en}{2021 {Asia}-{Pacific} {International}
  {Symposium} on {Aerospace} {Technology}}}, Jeju, Republic of Korea, Nov.
  2021.

\bibitem{px4}
\BIBentryALTinterwordspacing
PX4, ``{PX4} project,'' 2014. [Online]. Available: \url{http://px4.io}
\BIBentrySTDinterwordspacing

\bibitem{achtelikDesignMultiRotor2012}
M.~Achtelik, K.-M. Doth, D.~Gurdan, and J.~Stumpf,
  ``\BIBforeignlanguage{en}{Design of a {Multi} {Rotor} {MAV} with regard to
  {Efficiency}, {Dynamics} and {Redundancy}},'' in
  \emph{\BIBforeignlanguage{en}{{AIAA} {Guidance}, {Navigation}, and {Control}
  {Conference}}}, Minneapolis, MN, Aug. 2012.

\bibitem{domahidiECOSSOCPSolver2013}
A.~Domahidi, E.~Chu, and S.~Boyd, ``\BIBforeignlanguage{en}{{ECOS}: {An} {SOCP}
  {Solver} for {Embedded} {Systems}},'' in
  \emph{\BIBforeignlanguage{en}{European {Control} {Conference} ({ECC})}},
  Zurich, Switzerland, Jul. 2013.

\end{thebibliography}

%                                   % the textheight of the last page by a suitable amount.
%                                   % This command does not take effect until the next page
%                                   % so it should come on the page before the last. Make
%                                   % sure that you do not shorten the textheight too much.
%
%%%%%%%%%%%%%%%%%%%%%%%%%%%%%%%%%%%%%%%%%%%%%%%%%%%%%%%%%%%%%%%%%%%%%%%%%%%%%%%%

%%%%%%%%%%%%%%%%%%%%%%%%%%%%%%%%%%%%%%%%%%%%%%%%%%%%%%%%%%%%%%%%%%%%%%%%%%%%%%%%


% \section*{ACKNOWLEDGMENT}
% %%%%%%%%%%%%%%%%%%%%%%%%%%%%%%%%%%%%%%%%%%%%%%%%%%%%%%%%%%%%%%%%%%%%%%%%%%%%%%%%
%

\end{document}